\documentclass{fundam}

\publyear{22}
\papernumber{2102}
\volume{185}
\issue{1}

\publyear{22}
\papernumber{2102}
\volume{185}
\issue{1}

\usepackage{url} 
\usepackage[ruled,lined]{algorithm2e}
\usepackage{graphicx}
\usepackage{tikz}
\usetikzlibrary{automata, positioning, arrows}

\usepackage{amssymb,amsmath}
\usepackage{microtype}
\usepackage[small,basic]{complexity}
\usepackage{xspace}
\usepackage{mathtools}

\usepackage{bm}
\usepackage{cite}
\usepackage{enumerate}
\usepackage{thmtools}

\newtheorem{problem}{Problem}

\newenvironment{proofof}[1]
  {\trivlist\PRstyle\item[]{\bfseries Proof of #1:}\newline}{\QED\endtrivlist}

\renewcommand{\R}{\mathbb{R}}
\renewcommand{\C}{\mathbb{C}}
\newcommand{\GR}{\mathbb{Q}(\iu)}
\newcommand{\Q}{\mathbb{Q}}
\newcommand{\Z}{\mathbb{Z}}
\newcommand{\N}{\mathbb{N}}

\newcommand{\gltwoq}{{\rm GL}\ensuremath{(2,\mathbb{Q})}\xspace}

\newcommand{\slthreez}{{\rm SL}\ensuremath{(3,\mathbb{Z})}\xspace}

\newcommand{\heisc}{{\rm H}\ensuremath{(n,\mathbb{Q}(\iu))}\xspace}
\newcommand{\heiscomplex}{{\rm H}\ensuremath{(n,\mathbb{C})}\xspace}
\newcommand{\heisk}{{\rm H}\ensuremath{(n,\mathbb{K})}\xspace}
\renewcommand{\Re}{\operatorname{Re}}
\renewcommand{\Im}{\operatorname{Im}}
\newcommand{\veca}{\bm{a}}
\newcommand{\vecb}{\bm{b}}
\newcommand{\vecz}{\bm{0}}
\newcommand{\vect}{\textrm{vec}} 

\newcommand{\iu}{\mathrm{i}}

\DeclareMathOperator{\Arg}{arg}
\newcommand{\shuffle}{\mathsf{shuffle}}
\newcommand{\gammaeq}{\stackrel{\mathclap{\Arg}}{=}}

\newcommand{\bx}{\bm{x}}

\begin{document}

\title{On the Identity and Group Problems for Complex Heisenberg Matrices}

\address{R. Niskanen: r.niskanen{@}ljmu.ac.uk}

\author{Paul C. Bell\\
School of Computer Science and Mathematics, Keele University, UK \\
p.c.bell{@}keele.ac.uk
\and Reino Niskanen \\
School of Computer Science and Mathematics, Liverpool John Moores University, UK\\
r.niskanen{@}ljmu.ac.uk
\and Igor Potapov \\
Department of Computer Science, University of Liverpool, UK\\
potapov{@}liverpool.ac.uk
\and Pavel Semukhin \\
School of Computer Science and Mathematics, Liverpool John Moores University, UK\\
p.semukhin{@}ljmu.ac.uk
}

\maketitle

\runninghead{P. Bell, et al.}{On the Identity and Group Problems for Complex Heisenberg Matrices}

\begin{abstract}
We study the  Identity Problem, the problem of determining if a finitely generated semigroup  of matrices contains the identity matrix; see
Problem 3 (Chapter 10.3) in ``Unsolved Problems in Mathematical Systems and Control Theory'' by Blondel and Megretski (2004). This fundamental problem is  known to be
undecidable for  $\mathbb{Z}^{4 \times 4}$ and decidable for  $\mathbb{Z}^{2 \times 2}$. 
The Identity Problem has been recently shown to be in polynomial time by Dong for the Heisenberg group over complex numbers in any fixed dimension with the use of Lie algebra and the Baker-Campbell-Hausdorff formula.
We develop alternative proof techniques for the problem making a step forward towards more general problems such as the Membership Problem.
Using our techniques we also show that the problem of determining if a given set of Heisenberg matrices generates a group can be decided in polynomial time.
\end{abstract}


\section{Introduction}
Matrices and matrix products can represent dynamics in many systems, from computational applications in linear algebra and engineering to natural science applications in quantum mechanics, population dynamics and statistics, among others \cite{BM04,BN16,COW13,COW16,ding2015,GOW15,KLZ16,MT17,OPW15}.
The analysis of various evolving systems requires solutions of reachability questions in linear systems, which form the essential part of verification procedures, control theory questions, biological systems predictability, security analysis etc.

Reachability problems for matrix products are  challenging due to the complexity of this mathematical object and a lack of effective algorithmic techniques.   The significant challenge in the analysis of matrix semigroups was initially illustrated by Markov (1947) \cite{Markov47} and later highlighted by  Paterson (1970) \cite{Paterson70}, Blondel and Megretski (2004) \cite{BM04}, and Harju (2009) \cite{Harju09}. The central reachability question is  the \emph{Membership Problem}.
\begin{problem}[Membership Problem]\label{prob:membership}
Let $S$ be a matrix semigroup generated by a finite set of
$n{\times}n$ matrices over a ring $R$, where $R=\Z,\mathbb{Q},\mathbb{A},\mathbb{Q}(\iu)$, etc. 
Let \(M\) be an \(n\)-by-\(n\) matrix over the same ring. Is $M$ in the semigroup, i.e., does $M\in S$ hold?
\end{problem}
By restricting~$M$ to be the identity matrix, the problem is known as  the {\em Identity Problem}. 
\begin{problem}[Identity Problem]\label{idProb}
Let $S$
be a matrix semigroup generated by a finite set of
$n{\times}n$
matrices over a ring $R$, where $R=\Z,\mathbb{Q},\mathbb{A},\mathbb{Q}(\iu)$, etc. Is the identity matrix $\bm{I}$ in the semigroup, i.e., does $\bm{I}\in S$ hold?
\end{problem}

The Membership Problem is known to be undecidable for integer matrices starting from dimension three, but the 
decidability status of the Identity Problem was unknown for a long time for matrix semigroups of any dimension;
see Problem 10.3 in ``Unsolved Problems in Mathematical Systems and Control Theory'' \cite{BM04}. The Identity Problem was shown to be undecidable for $48$ matrices from $\mathbb{Z}^{4 \times 4}$ in \cite{BP10} and  for a generating set of eight matrices in \cite{KNP18}. This implies that the \emph{Group Problem} (decide whether a finitely generated semigroup is a group) is also undecidable \cite{BP10}. The Identity Problem and the Group Problem are open for $\mathbb{Z}^{3 \times 3}$.

The Identity Problem for a semigroup generated by $2 \times 2$ 
 matrices was shown to be  $\EXPSPACE$ decidable in
 \cite{CK05} and later improved by showing to be $\NP$-complete in  \cite{BHP17}. The only decidability results beyond integer $2\times2$ matrices were shown in \cite{DPS20}  
 for flat rational subsets of \gltwoq .

Similarly to \cite{CHK99}, the work  \cite{KNP18} initiated  consideration of matrix decision problems  in  the Special Linear Group \slthreez, by showing that
there is no embedding from pairs of words into matrices from \slthreez.
The authors also proved that the Identity Problem is decidable for the discrete Heisenberg group ${\rm H}(3,\mathbb{Z})$ which is a subgroup of \slthreez. 

The Heisenberg group is widely used in mathematics and physics.
This is in some sense the simplest non-commutative group, and has close
connections to  quantum mechanical systems~\cite{Brylinski93,GU14,Kostant70}, harmonic analysis, and  number theory \cite{Persi2017,Persi2021}. It also  makes appearances in complexity theory, e.g., the analysis and geometry of the Heisenberg group have been used to disprove the Goemans-Linial conjecture in complexity theory \cite{FOCS2006}.
Matrices in physics and engineering are ordinarily defined with values over \(\mathbb{R}\) or \(\mathbb{C}\).
In this context, we formulate our decision problems and  algorithmic solutions over the field of complex numbers with a finite representation, namely, the Gaussian rationals  \(\GR\).

There has been an increased interest in the Identity Problem for different algebraic objects that encompass the Heisenberg group.
The Heisenberg group has nilpotency class two.
Dong~\cite{Dong24soda} investigated the Identity Problem in larger nilpotent subgroups, proving that if the nilpotency class is at most ten then the Identity Problem is decidable in polynomial time.
Independently, Shafrir~\cite{Shafrir24} showed that the Identity Problem is decidable for all nilpotent subgroups.
Dong~\cite{Dong24stoc} also studied another group property that Heisenberg groups (and all groups of nilpotency class at most three) satisfy.
Namely, they showed that the Identity Problem is decidable for metabelian groups. This result was subsequently extended in Bodart and Dong \cite{BodartDong24}, where it was shown that the Identity Problem is decidable in virtually solvable matrix groups over the field of algebraic numbers. Moreover, the Identity Problem was shown to be \(\NP\)-complete in the special affine group of \(\Z^2\); see Dong \cite{DongLICS23}.

In another resent paper \cite{Dong23}, Dong showed that the Identity Problem is decidable in polynomial time for complex Heisenberg matrices.
They first prove the result for upper-triangular matrices with rational entries and ones on the main diagonal, \(\mathrm{UT}(\Q)\), and then use a known embedding of the Heisenberg group over algebraic numbers into \(\mathrm{UT}(\Q)\). Their approach is different from our techniques; the main difference being that \cite{Dong23} uses tools from Lie algebra and, in particular, matrix logarithms and the Baker-Campbell-Hausdorff formula, to reason about matrix products and their properties. In contrast, our approach first characterises matrices which are `close to' the identity matrix, which we denote \(\Omega\)-matrices. Such matrices are close to the identity matrix in that they differ only in a single position in the top-right corner. We then argue about the commutator angle of matrices within this set in order to determine whether zero can be reached, in which case the identity matrix is reachable. We believe that these techniques take a step towards proving the decidability of the more general \emph{Membership Problem}, which we discuss at the end of the paper. A careful analysis then follows to ensure that all steps require only polynomial time, and we extend our techniques to show that determining if a given set of matrices forms a group (the \emph{Group Problem}) is also decidable in \(\P\) (this result is shown in \cite{Dong24soda} using different techniques). We thus present polynomial time algorithms for both these problems for Heisenberg matrices over $\mathbb{Q}$(i) in any dimension $n$. 

 These new techniques allow us to extend previous results for the discrete Heisenberg group H($n,\mathbb{Z}$) and H($n,\mathbb{Q}$) 
\cite{KNP18,COS+19,KLZ16,Dong22} and make a step forward towards proving the decidability of the Membership Problem for complex Heisenberg matrices.
This paper is an extended version of the conference paper \cite{BNP+23} containing full proofs of the results and extended discussion of the proof techniques.
\section{Roadmap}

We will give a brief overview of our approach here. Given a Heisenberg matrix \(M=\begin{psmallmatrix}
1&\bm{m}_1^T&m_3\\\bm{0}&\bm{I}_{n-2}&\bm{m}_2\\0&\bm{0}^T&1
\end{psmallmatrix} \in \heisc\), denote by $\psi(M)$ the triple~$(\bm{m}_1,\bm{m}_2,m_3) \in \GR^{2n-3}$. We
define the set \(\Omega\subseteq\heisc\) as those matrices where \(\bm{m}_1\) and \(\bm{m}_2\) are zero vectors, i.e., matrices in \(\Omega\) look like \(\bm{I}_n\) except allowing any element of \(\GR\) in the top right element. Such matrices play a crucial role in our analysis.

In particular, given a set of matrices \(G = \{G_1, \ldots, G_t\} \subseteq \heisc\) generating a semigroup \(\langle G \rangle\), we can find a description of \(\Omega_{\langle G \rangle} = \langle G \rangle \cap \Omega\). Since \(\bm{I} \in \Omega\), the Identity Problem reduces to determining if \(\bm{I} \in \Omega_{\langle G \rangle}\).

Several difficulties present themselves, particularly if we wish to solve the problem in polynomial time (\(\P\)). The set \(\Omega_{\langle G \rangle}\) is described by a linear set \(\mathcal{S} \subseteq \mathbb{N}^t\), which is the solution set of a homogeneous system of linear Diophantine equations induced by matrices in \(G\).
This is due to the observation that the elements \((\bm{m}_1,\bm{m}_2) \in \GR^{2n-4}\) behave in an additive fashion under multiplication of Heisenberg matrices. The main issue is that the size of the basis of \(\mathcal{S}\) is exponential in the description size of \(G\). Nevertheless, we can determine \emph{if a solution exists} to such a system in \(\P\) (Lemma~\ref{Ptime}), and this proves sufficient. 

The second issue is that reasoning about the element \(m_3 \in \GR\) (i.e., the top right element) in a product of Heisenberg matrices is much more involved than for elements \((\bm{m}_1,\bm{m}_2) \in \GR^{2n-4}\). Techniques to determine if \(m_3 = 0\) for an \(\Omega\)-matrix within \(\Omega_{\langle G \rangle}\) take up the bulk of this paper. 

The key to our approach is to consider \emph{commutators} of pairs of matrices within \(G\), which in our case can be described by a single complex number.
We denote the commutator of matrices \(M_1\), \(M_2\) by \([M_1,M_2]\); see the following section for the definition of the commutator.
After 
removing all \emph{redundant matrices} (those never reaching an \(\Omega\)-matrix), we have two cases to consider. Either every pair of matrices from \(G\) has the same \emph{angle} of the commutator or else there are at least two commutators with different angles.

The latter case is used in Lemma~\ref{lem:notallsamegamma}.
It states that the identity matrix can always be constructed using a solution that contains four particular matrices.
Let \(M_{1}\), \(M_{2}\), \(M_{3}\) and \(M_{4}\) be such that \([M_{1},M_{2}]=r\exp(\iu \gamma)\) and \([M_{3},M_{4}]=r'\exp(\iu \gamma')\), where \(\gamma\neq \gamma'\) so that pairs \(M_{1}, M_{2}\) and \(M_{3}, M_{4}\) have different commutator angles.
We may then define four matrix products using the same generators but matrices \(M_1\), \(M_2\), \(M_3\) and \(M_4\) are in a different order. 
This difference in order and the commutator angles being different, ensures that we can control the top right corner elements in order to construct the identity matrix.
Lemma~\ref{lem:purewithcommutators} provides details on how to calculate the top right element in these products.
We then prove that these top right elements in the four matrices are not contained in an open half-plane and this is sufficient for us to construct the identity matrix.

The above construction does not work when all commutators have the same angle, and indeed in this case the identity matrix may or may not be present.
Hence, we need to consider various possible shuffles of matrices in these products.
To this end, we extend the result of Lemma~\ref{lem:purewithcommutators} to derive a formula for the top right element for any shuffle and prove it as Lemma~\ref{lem:anyshuffle}.
We observe that there is a \emph{shuffle invariant} part of the product that does not depend on the shuffle, and that shuffles add or subtract commutators.
Furthermore, this shuffle invariant component can be calculated from the generators used in the product.
As we assume that all commutators have the same angle, \(\gamma\), different shuffles move the value along the line in the complex plane defined by the common commutator angle which we call the \emph{\(\gamma\)-line}.

It is straightforward to see that if it is not possible to reach the \(\gamma\)-line using the additive semigroup of shuffle invariants, then the identity matrix cannot be generated.
Indeed, since different shuffles move the value along the \(\gamma\)-line but the shuffle invariant part never reaches it, then the possible values are never on the \(\gamma\)-line, which includes the origin.

We show that if it \emph{is} possible to reach the \(\gamma\)-line using shuffle invariants and there are at least two non-commuting matrices in the used solution, then the identity matrix is in the semigroup (Lemma~\ref{lem:samegamma}). Testing this property requires determining the solvability of a polynomially-sized set of non-homogeneous systems of linear Diophantine equations, which can be done in polynomial time by Lemma~\ref{Ptime}.

Finally, if the \(\gamma\)-line can be reached only with commuting matrices, then we can decide if the identity matrix is in the semigroup by solving another system of linear Diophantine equations. This system is constructed using an explicit formula for the top right element in terms of the Parikh vector of the generators (see Lemma~\ref{lem:samegamma}).

\section{Preliminaries}
The sets of rational, real and complex numbers are denoted by $\Q$, \(\R\) and \(\C\), respectively.
The set of rational complex numbers is denoted by $\GR=\{a+b\iu\mid a,b\in \Q\}$.
The set \(\GR\) is often called the Gaussian rationals in the literature.
A complex number can be written in polar form  $a+b\iu=r\exp(\iu\varphi)$, where $r\in\R$ and $\varphi\in [0,\pi)$.
We denote the \emph{angle} of the polar form \(\varphi\) by \(\Arg(a+b\iu)\). 
We also denote \(\Re(a+b\iu)=a\) and \(\Im(a+b\iu)=b\).
It is worth highlighting that commonly the polar form is defined for a positive real \(r\) and an angle in \([0,2\pi)\).
These two forms are equivalent in the sense that one can be easily transformed into the other.

The \emph{identity matrix} is denoted by $\bm{I}_n$ or, if the dimension \(n\) is clear from the context, by \(\bm{I}\).
Let \(\mathbb{K}\) be any field. The \emph{Heisenberg group} \heisk over \(\mathbb{K}\) is formed by $n \times n$ matrices of the form 
$M  = 
\begin{psmallmatrix}
1 & \bm{m}_1^T & m_3\\
\bm{0} & \bm{I}_{n-2} & \bm{m_2}\\
0 & \bm{0}^T & 1
\end{psmallmatrix}$, 
where \(\bm{m}_1,\bm{m}_2 \in \mathbb{K}^{n-2}\), \(m_3 \in \mathbb{K}\) and \(\bm{0} = (0, 0, \ldots, 0)^T \in \mathbb{K}^{n-2}\) is the zero vector.
Note that the Heisenberg group is a non-commutative subgroup of ${\rm SL}(n,\mathbb{K})=\{M\in \mathbb{K}^{n\times n}\mid \det(M)=1\}$. 

We will be interested in subsemigroups of \heisc which are finitely generated. Given a set of matrices \(G = \{G_1, \ldots, G_t\} \subseteq \heisc\), we denote the matrix semigroup  generated by \(G\) as \(\langle G \rangle\).

Let \(M=\begin{psmallmatrix}
1&\bm{m}_1^T&m_3\\\bm{0}&\bm{I}_{n-2}&\bm{m}_2\\0&\bm{0}^T&1
\end{psmallmatrix}\), then \((M)_{1,n} = m_3\) is the top right element.
To improve readability, by $\psi(M)$ we denote the triple~$(\bm{m}_1,\bm{m}_2,m_3) \in \GR^{2n-3}$. 

The vectors \(\bm{m}_1,\bm{m}_2\) play a crucial role in our considerations.
We define the set \(\Omega\subseteq\heisc\) as those matrices where \(\bm{m}_1\) and \(\bm{m}_2\) are zero vectors, i.e., matrices in \(\Omega\) look like \(\bm{I}_n\) except allowing any element of \(\GR\) in the top right element.
That is,
\(\Omega=\left\{\begin{psmallmatrix}
1&\bm{0}^T&m_3\\\bm{0}&\bm{I}_{n-2}&\bm{0}\\0&\bm{0}^T&1
\end{psmallmatrix} \mid m_3\in \GR \right\}\),
where \(\bm{0} = (0, 0, \ldots, 0)^T \in \GR^{n-2}\) is the zero vector.

Let us define a shuffling of a product of matrices.\footnote{In this paper, the term ``product of matrices \(M_{1},M_2,\ldots, M_{k}\)'' can either mean the expression \(M_{1}M_{2}\cdots M_{k}\) or the resulting matrix \(M=M_{1}M_{2}\cdots M_{k}\). It should be clear from the context which meaning is used.} Let \(M_{1},M_2,\ldots, M_{k} \in \langle G\rangle\). The set of permutations of a product of these matrices is denoted by 
\[\shuffle(M_{1},M_2,\ldots, M_{k})=\{M_{\sigma(1)}M_{\sigma(2)}\cdots M_{\sigma(k)}\mid \sigma \in \mathcal{S}_k\},\]
where \(\mathcal{S}_k\) is the set of permutations on \(k\) elements.
If some matrix appears multiple times in the list, say \(M_1\) appears \(x\) times, we write 
\(\shuffle(M_1^x,M_2,\ldots,M_k)\) instead of \(\shuffle(\underbrace{M_1,\ldots,M_1}_{x \text{ times}},M_2,\ldots,M_k)\).
Furthermore, for a matrix \(M=M_{1}M_{2}\cdots M_{k}\in\langle G\rangle\), we define
\[
\shuffle(M)=\shuffle(M_{1}M_{2}\cdots M_{k}) \text{ to be equal to }  \shuffle(M_{1},M_{2},\ldots, M_{k}).    
\]

Let \(M_1=\begin{psmallmatrix}1&\veca_1^T&c_1\\\bm{0}&\bm{I}_{n-2}&\vecb_1\\0&\bm{0}^T&1\end{psmallmatrix}\) and \(M_2=\begin{psmallmatrix}1&\veca_2^T&c_2\\\bm{0}&\bm{I}_{n-2}&\vecb_2\\0&\bm{0}^T&1\end{psmallmatrix}\).
By an abuse of notation, we define the \emph{commutator} \([M_1,M_2]\) of \(M_1\) and \(M_2\) by \([M_1,M_2]=\veca_1^T\vecb_2-\veca_2^T\vecb_1\in\GR\). Note that the commutator of two arbitrary matrices $A, B$ is ordinarily defined as \([A, B] = AB - BA\), i.e., a matrix. However, for matrices \(M_1, M_2 \in \heisc\), it is clear that \(M_1M_2-M_2M_1 = \begin{psmallmatrix}0&\vecz^T&\veca_1^T\vecb_2-\veca_2^T\vecb_1\\\bm{0}&O&\vecz\\0&\bm{0}^T&0\end{psmallmatrix}\), where \(O\) is the \((n-2)\times (n-2)\) zero matrix,
thus justifying our notation which will be used extensively. Observe that the matrices \(M_1, M_2\) commute if and only if \([M_1,M_2]=0\).

Note that the commutator is antisymmetric, i.e., \([M_1,M_2]=-[M_2,M_1]\).
We further say that \(\gamma\) is the \emph{angle of the commutator} if \([M_1,M_2]=r\exp(\iu\gamma)\) for some \(r\in\R\) and \(\gamma\in[0,\pi)\).
If two commutators \([M_1,M_2]\), \([M_3,M_4]\) have the same angles, that is, \([M_1,M_2]=r\exp(\iu\gamma)\) and \([M_3,M_4]=r'\exp(\iu\gamma)\) for some \(r,r'\in\R\), then we denote this property by \([M_1,M_2]\gammaeq[M_3,M_4]\). If they have different angles, then we write \([M_1,M_2]\not\gammaeq[M_3,M_4]\).
 By convention, if \([M_1,M_2]=0\), then \([M_1,M_2]\gammaeq[M_3,M_4]\) for every \(M_3, M_4 \in \heisc\).

To show that our algorithms run in polynomial time, we will need the following lemma.

\begin{lemma}\label{Ptime}
The following problems can be solved in polynomial time:
    \begin{enumerate}[(i)]
        \item Let \(A\in \Q^{n\times m}\) be a rational matrix, and \(\bm{b}\in \Q^n\) be an \(n\)-dimensional rational vector with non-negative coefficients. Decide whether the system of inequalities \(A\bm{x} \geq \bm{b}\) has an integer solution \(\bm{x}\in \Z^m\).

        \item Let \(A_1\in \Q^{n_1\times m}\) and \(A_2\in \Q^{n_2\times m}\) be a rational matrices. Decide whether the system of inequalities \(A_1\bm{x} \geq \bm{0}^{n_1}\) and \(A_2\bm{x} > \bm{0}^{n_2}\) has an integer solution \(\bm{x}\in \Z^m\).
    \end{enumerate}
\end{lemma}
\begin{proof} 
Let us prove both cases individually.

    \emph{(i)} We will show that the system \(A\bm{x} \geq \bm{b}\) has an integer solution \(\bm{x}\in \Z^m\) if and only if it has a rational solution \(\bm{x}\in \Q^m\). One direction is obvious. So, suppose there is a rational vector \(\bm{x}\in \Q^m\) such that \(A\bm{x} \geq \bm{b}\). Let \(r\geq 1\) be the least common multiple of the denominators of all the coefficients of \(\bm{x}\). Let \(\bm{x'} = r\bm{x}\in \Z^m\). Hence we have
    \[
    A\bm{x'} = A(r\bm{x}) = r(A\bm{x}) \geq r\bm{b} \geq \bm{b},
    \]
    where the last two inequalities hold because \(\bm{b}\) has non-negative coefficients and \(r\geq 1\).
    To finish the proof, note that we can decide in polynomial time whether a system of linear inequalities has a rational solution using linear programming; see \cite{Schrijver98}.

    \emph{(ii)} Since the system of inequalities is homogeneous, we can assume without loss of generality that matrices \(A_1\), \(A_2\) have integer coefficients. Hence the condition \(A_2\bm{x} > \bm{0}^{n_2}\) is equivalent to \(A_2\bm{x} \geq \bm{1}^{n_2}\) for any integer vector  \(\bm{x}\in \Z^m\), where \(\bm{1}^{n_2}\) is a vector of dimension \(n_2\) with coordinates \(1\). By the first part, we can decide in polynomial time whether the system of inequalities \(A_1\bm{x} \geq \bm{0}^{n_1}\) and \(A_2\bm{x} \geq \bm{1}^{n_2}\) has an integer solution \(\bm{x}\in \Z^m\). 
\end{proof}

\section{Properties of \(\Omega\)-matrices}\label{SecShufInv}

To solve the Identity Problem for subsemigroups of \(\heisc\) (Problem~\ref{idProb}), we will be analysing matrices in \(\Omega\) (matrices with all zero elements, except possibly the top-right corner value). 
Let us first discuss how to construct \(\Omega\)-matrices from a given set of generators \(G \subseteq \heisc\).

As observed earlier, when multiplying Heisenberg matrices of the form \(\begin{psmallmatrix}
1&\bm{m}_1^T&m_3\\\bm{0}&\bm{I}_{n-2}&\bm{m}_2\\0&\bm{0}^T&1
\end{psmallmatrix}\),
elements \(\bm{m}_1\) and \(\bm{m}_2\) are \emph{additive}.
We can thus construct a homogeneous system of linear Diophantine equations (SLDEs) induced by matrices in \(G\).
Each \(\Omega\)-matrix in \(\Omega_{\langle G \rangle}\) then corresponds to a solution to this system.

Let \(G=\{G_1,\ldots,G_t\}\), where \(\psi(G_i)=(\bm{a}_i,\bm{b}_i,c_i)\). For a vector \(\veca \in \GR^{n-2}\), define a column vector \(\Re(\veca) = (\Re(\veca(1)), \ldots, \Re(\veca(n-2)))^T\) (and similarly for \(\Im(\veca)\)).
We consider a system \(A\bx=\bm{0}\), where
\begin{equation}\label{eq:matA}
A=\begin{pmatrix}
\Re(\veca_1) & \Re(\veca_2) & \cdots & \Re(\veca_t) \\
\Im(\veca_1) & \Im(\veca_2) & \cdots & \Im(\veca_t) \\
\Re(\vecb_1) & \Re(\vecb_2) & \cdots & \Re(\vecb_t) \\
\Im(\vecb_1) & \Im(\vecb_2) & \cdots & \Im(\vecb_t)
\end{pmatrix},
\end{equation}
\(\bx\in\N^t\) and \(\bm{0}\) is the \(4(n-2)\)-dimensional zero vector; noting that \(A \in \mathbb{Q}^{4(n-2) \times t}\).
Let \(\mathcal{S}=\{\bm{s}_1,\ldots,\bm{s}_p\}\) be the set of \emph{minimal} or \emph{irreducible} solutions to the system,
that is, the solutions that cannot be written as a sum of two nonzero solutions.
The set \(\mathcal{S}\) is always finite and constructable  \cite{Schrijver98}.

A matrix \(G_i\in G\) is \emph{redundant} if the \(i\)th component is 0 in every minimal solution \(\bm{s}\in\mathcal{S}\).
Non-redundant matrices can be recognized by checking whether a non-homogeneous SLDE has a solution. 
More precisely, to check whether \(G_i\) is non-redundant, we consider the system \(A\bx=\bm{0}\) together with the constraint that \(\bx(i)\geq1\), where \(\bx(i)\) is the \(i\)th component of \(\bx\). Using Lemma~\ref{Ptime}, we can determine in polynomial time whether such a system has an integer solution.

For the remainder of the paper, we assume that \(G\) is the set of non-redundant matrices because redundant matrices cannot be used to generate an \(\Omega\)-matrix, and hence to generate \(\bm{I}\). 
This implicitly assumes that for this \(G\), the set \(\mathcal{S}\neq\emptyset\).
Indeed, if \(\mathcal{S} = \emptyset\), then all matrices are redundant and \(G\) must be the empty set. Hence \(\bm{I}\not\in \langle G\rangle\) holds trivially in this case.

Let \(M_{1},\ldots,M_{k}\in G\) be such that \(X = M_{1} \cdots  M_{k} \in \Omega\).
The Parikh vector of occurrences of each matrix from \(G\) in product \(X\) may be written as \(\bm{x} = (m_1, \ldots, m_{t}) \in \mathbb{N}^{t}\).
This Parikh vector \(\bm{x}\) is a linear combination of elements of \(\mathcal{S}\), i.e., 
\(
\bm{x} = \sum_{j=1}^p y_j \bm{s}_j
\), with \(y_j \in \mathbb{N}\), because  \(\bm{x}\) is a solution to the SLDEs.
Each element of \(\shuffle(M_1, \ldots, M_k)\) has the same Parikh vector, but their product is not necessarily the same matrix; potentially differing in the top right element.

Let us state some properties of \(\Omega\)-matrices.
\begin{lemma}\label{omegaPropLem}
The \(\Omega\)-matrices are closed under matrix product; the top right element is additive under the product of two matrices; and \(\Omega\)-matrices commute with Heisenberg matrices.
In other words, let \(A, B \in \Omega\) and \(M\in\heisc\), then 
\begin{enumerate}[(i)]
\item \(AB \in \Omega\);
\item \((AB)_{1,n} = A_{1,n} + B_{1,n}\);
\item \(AM=MA\).
\end{enumerate}
Furthermore, if \(N=M_1M_2\cdots M_{k-1}M_k \in \Omega\) for some \(M_1,\ldots,M_k\in\heisc\), then every cyclic permutation of matrices results in \(N\).
That is,
\[N=M_2M_3\cdots M_kM_1=\cdots= M_kM_1\cdots M_{k-2}M_{k-1}.\]
\end{lemma}

\begin{proof}
The first three claims follow from the definition of \(\Omega\)-matrices.
Let us present the proofs for the sake of completeness.
Let \(A=\begin{psmallmatrix}1&\bm{0}^T&a_3\\\bm{0}&\bm{I}_{n-2}&\bm{0}\\0&\bm{0}^T&1\end{psmallmatrix}\) and \(B=\begin{psmallmatrix}1&\bm{0}^T&b_3\\\bm{0}&\bm{I}_{n-2}&\bm{0}\\0&\bm{0}^T&1\end{psmallmatrix}\).
Now, \(AB=\begin{psmallmatrix}1&\bm{0}^T&a_3+b_3\\\bm{0}&\bm{I}_{n-2}&\bm{0}\\0&\bm{0}&1\end{psmallmatrix}\) as was claimed in (i) and (ii).
Let \(M=\begin{psmallmatrix}1&\bm{m}_1^T&m_3\\\bm{0}&\bm{I}_{n-2}&\bm{m}_2\\0&\bm{0}^T&1\end{psmallmatrix}\); then we have
\begin{align*}
    AM=\begin{psmallmatrix}1&\bm{m}_1^T&a_3+m_3\\\bm{0}&\bm{I}_{n-2}&\bm{m}_2\\0&\bm{0}^T&1\end{psmallmatrix} \text{ and } MA=\begin{psmallmatrix}1&\bm{m}_1^T&m_3+a_3\\\bm{0}&\bm{I}_{n-2}&\bm{m}_2\\0&\bm{0}^T&1\end{psmallmatrix},
\end{align*}
which proves (iii).

Let us proof the final claim.
We will show that \(M_1M_2\cdots M_{k-1}M_k= M_2M_3\cdots M_kM_1\). 
The other cyclic permutations are proven analogously. Note that (iii) is equivalent to \(M^{-1}AM = A\) for any \(A \in \Omega\) and \(M\in\heisc\). If \(N=M_1M_2\cdots M_{k-1}M_k \in \Omega\), then we have
\[
N = M_1^{-1} N M_1 = M_2M_3\cdots M_kM_1,
\]
which proves the claim.

\end{proof}

We require the following technical lemma that allows us to calculate the value in top right corner for particular products.
The claim is proven by a direct computation.
\begin{lemma}\label{lem:purewithcommutators}
Let \(M_{1},M_{2},\ldots,M_{k}\in \heisc\) such that \(M_{1}M_{2}\cdots M_{k} \in \Omega\) and let \(\ell\geq 1\).
Then, 
\[(M_{1}^\ell M_{2}^\ell\cdots M_{k}^\ell)_{1,n}= \ell \sum_{i=1}^k \left(c_i-\frac{1}{2}\veca_i^T\vecb_i\right)+\frac{\ell^2}{2}\sum_{1\leq i<j\leq k-1}[M_i,M_j],\]
where \(\psi(M_i)=(\bm{a}_i,\bm{b}_i,c_i)\) for each \(i=1,\ldots,k\).
\end{lemma}
\begin{proof}
Denote \(\psi(M_i)=(\veca_i,\vecb_i,c_i)\).
A direct calculation shows that the element in the top right corner is
\[\sum_{i=1}^k\ell c_i + \sum_{i=1}^k \frac{\ell(\ell-1)}{2}\veca_i^T\vecb_i+\ell^2\sum_{1\leq i<j\leq k}\veca_i^T\vecb_j.\]
That is, the coefficient of \(\ell\) is already in the desired form.
Let us rewrite the coefficient of \(\ell^2\) as follows:
\begin{align*}
    \sum_{i=1}^{k-1} \frac{1}{2}\veca_i^T\vecb_i &+ \frac{1}{2}\veca_k^T\vecb_k+\sum_{1\leq i<j\leq k-1}\veca_i^T\vecb_j+\sum_{i=1}^{k-1}\veca_i^T\vecb_k  \\
    &= \sum_{i=1}^{k-1} \frac{1}{2}\veca_i^T\vecb_i+\sum_{1\leq i<j\leq k-1}\veca_i^T\vecb_j+\sum_{i=1}^{k-1}\veca_i^T\vecb_k+\veca_k^T\vecb_k-\frac{1}{2}\veca_k^T\vecb_k  \\
    &= \sum_{i=1}^{k-1} \frac{1}{2}\veca_i^T\vecb_i+\sum_{1\leq i<j\leq k-1}\veca_i^T\vecb_j+\sum_{i=1}^{k}\veca_i^T\vecb_k-\frac{1}{2}\veca_k^T\vecb_k  \\
    &= \sum_{i=1}^{k-1} \frac{1}{2}\veca_i^T\vecb_i+\sum_{1\leq i<j\leq k-1}\veca_i^T\vecb_j-\frac{1}{2}\veca_k^T\vecb_k,
\end{align*}
where the first equality was obtained by taking terms with subindex \(k\) out of the sums, the second by adding \(\frac{1}{2}\veca_k^T\vecb_k-\frac{1}{2}\veca_k^T\vecb_k\), and the final two by observing that \(\sum_{i=1}^k\veca_i=\vecz\).
Now, \(\veca_k^T\vecb_k\) can be rewritten as
\[\veca_k^T\vecb_k=\left(\sum_{i=1}^{k-1}\veca_i\right)^T \cdot \sum_{j=1}^{k-1}\vecb_j = \sum_{1\leq i<j\leq k-1}\veca_i^T\vecb_j+\sum_{1\leq j<i\leq k-1}\veca_i^T\vecb_j+\sum_{i=1}^{k-1}\veca_i^T\vecb_i.\]
Combining this with the previous equation, we finally obtain that the coefficient of \(\ell^2\) is
\begin{align*}
    \sum_{i=1}^{k-1} \frac{1}{2}\veca_i^T\vecb_i &+ \sum_{1\leq i<j\leq k-1}\veca_i^T\vecb_j-\frac{1}{2}\left(\sum_{1\leq i<j\leq k-1}\veca_i^T\vecb_j+\sum_{1\leq j<i\leq k-1}\veca_i^T\vecb_j+\sum_{i=1}^{k-1}\veca_i^T\vecb_i\right) \\
    &= \frac{1}{2}\left(\sum_{1\leq i<j\leq k-1}\veca_i^T\vecb_j-\sum_{1\leq j<i\leq k-1}\veca_i^T\vecb_j\right)
    = \frac{1}{2}\sum_{1\leq i<j\leq k-1}[M_i,M_j].
\end{align*}
Thus completing the proof. 
\end{proof}

If we further assume that the matrices from the previous lemma commute, then for every \(M\in\shuffle(M_{1}^\ell,M_{2}^\ell,\ldots, M_{k}^\ell)\): 
\begin{equation}\label{commutingURval}
M_{1,n}  =  \ell \sum_{i=1}^k \left(c_i-\frac{1}{2}\veca_i^T\vecb_i\right)+\frac{\ell^2}{2}\sum_{1\leq i<j\leq k-1}[M_i,M_j] = \ell \sum_{i=1}^k \left(c_i-\frac{1}{2}\veca_i^T\vecb_i\right),
\end{equation}
noting that \([M_i,M_j] = 0\) when matrices \(M_i\) and \(M_j\) commute.

In Lemma~\ref{lem:purewithcommutators}, the matrix product has an ordering which yielded a simple presentation of the value in the top right corner.
In the next lemma, we consider an arbitrary shuffle of the product and show that the commutators are important when expressing the top right corner element.

\begin{lemma}\label{lem:anyshuffle}
Let \(M_{1},M_{2},\ldots,M_{k}\in \heisc\) such that \(M_{1}M_{2}\cdots M_{k} \in \Omega\) and let \(\ell\geq 1\).
Let \(M\) be a shuffle of the product \(M_1^\ell M_2^\ell\cdots M_k^\ell\) by a permutation \(\sigma\) that acts on \(k\ell\) elements.
Then
\[(M)_{1,n}=\ell \sum_{i=1}^k \left(c_i-\frac{1}{2}\veca_i^T\vecb_i\right)+\frac{\ell^2}{2}\sum_{1\leq i<j\leq k-1} [M_i,M_j] -\sum_{1\leq i<j\leq k}z_{ji}[M_i,M_j],\]
where \(\psi(M_i)=(\bm{a}_i,\bm{b}_i,c_i)\) for \(i=1,\ldots,k\), and \(z_{ji}\) is the number of times \(M_j\) appears before \(M_i\) in the product; so \(z_{ji}\) is the number of inversions of \(i,j\) in \(\sigma\).
\end{lemma}
\begin{proof}
As in Lemma~\ref{lem:purewithcommutators}, we proceed by a direct calculation.
Denote \(\psi(M_i)=(\veca_i,\vecb_i,c_i)\).
First, let us denote by \(z_{ij}\) the number of times matrix \(M_i\) is to the left of matrix \(M_j\) in \(M\).
Note, that \(z_{ij}+z_{ji}=\ell^2\) as there are in total \(\ell^2\) multiplications of \(M_i\) and \(M_j\).
Now, the direct calculation of the top right element in \(M\) is
\[\sum_{i=1}^k\ell c_i + \sum_{i=1}^k \frac{\ell(\ell-1)}{2}\veca_i^T\vecb_i+\sum_{1\leq i<j\leq k}z_{ij}\veca_i^T\vecb_j+\sum_{1\leq j<i\leq k}z_{ij}\veca_i^T\vecb_j.\]
As in the proof of Lemma~\ref{lem:purewithcommutators}, the \(\ell\) terms are as in the claim.
Let us focus on the \(\ell^2\) term.
We add \(\sum_{1\leq i<j\leq k}z_{ji}\veca_i^T\vecb_j-\sum_{1\leq i<j\leq k}z_{ji}\veca_i^T\vecb_j=0\) to the term, resulting in
\begin{align*}
    \ell^2\sum_{i=1}^k \frac{1}{2}\veca_i^T\vecb_i &+\sum_{1\leq i<j\leq k}z_{ij}\veca_i^T\vecb_j+\sum_{1\leq j<i\leq k}z_{ij}\veca_i^T\vecb_j+\sum_{1\leq i<j\leq k}z_{ji}\veca_i^T\vecb_j-\sum_{1\leq i<j\leq k}z_{ji}\veca_i^T\vecb_j \\
    &= \ell^2\sum_{i=1}^k \frac{1}{2}\veca_i^T\vecb_i+\ell^2\sum_{1\leq i<j\leq k}\veca_i^T\vecb_j+\sum_{1\leq j<i\leq k}z_{ij}(\veca_i^T\vecb_j-\veca_j^T\vecb_i)\\
    &=\frac{\ell^2}{2}\sum_{1\leq i<j\leq k-1}[M_i,M_j]-\sum_{1\leq i<j\leq k}z_{ji}[M_i,M_j].
\end{align*}
In the above equation the first equality follows from the fact that \(z_{ij}+z_{ji}=\ell^2\) and the second equality can be obtained via analogous calculation as in the proof of Lemma~\ref{lem:purewithcommutators}. 
\end{proof}

The crucial observation is that regardless of the shuffle, the top right corner element has a common term, namely \(\sum_{i=1}^k (c_i-\frac{1}{2}\veca_i^T\vecb_i)\), plus some linear combination of commutators.
We call the common term the \emph{shuffle invariant}. 
Note that the previous lemmas apply to any Heisenberg matrices, even those in \(\heiscomplex\). For the remainder of the section, we restrict considerations to matrices in \(G\).

\begin{definition}[Shuffle Invariant]\label{shufinvdef}
Let \(M_{1},\ldots,M_{k}\in G\) be such that \(X = M_{1}\cdots M_{k} \in \Omega\). The Parikh vector of occurrences of each matrix from \(G\) in product \(X\) may be written as \(\bm{x} = (m_1,  \ldots, m_{t}) \in \mathbb{N}^{t}\) where \(t = |G|\) as before. Define \(\Lambda_{\bm{x}} = 
\sum_{i=1}^{t}m_i(c_i-\frac{1}{2}\veca_i^T\vecb_i)\) as the shuffle invariant of Parikh vector~\(\bm{x}\).
\end{definition}

Note that the shuffle invariant is dependant only on the generators used in the product and the Parikh vector \(\bm{x}\).

Let \(\mathcal{S} = \{\bm{s}_1, \ldots, \bm{s}_p\} \subseteq \mathbb{N}^k\) be the set of minimal solutions to the system of linear Diophantine equations for \(G\) giving an \(\Omega\)-matrix, as described in the beginning of the section.
Each \(\bm{s}_j\) thus induces a shuffle invariant that we denote \(\Lambda_{\bm{s}_j} \in \GR\) as shown in Definition~\ref{shufinvdef}. The Parikh vector of any \(X = M_{1}M_{2}\cdots M_{k}\) with \(X \in \Omega\), denoted \(\bm{x}\), is a linear combination of elements of \(\mathcal{S}\), i.e., 
\(
\bm{x} = \sum_{j=1}^p y_j \bm{s}_j
\). 
We then note that the shuffle invariant \(\Lambda_{\bm{x}}\) of \(\bm{x}\) is 
\[
\Lambda_{\bm{x}} = \sum_{j=1}^p y_j \Lambda_{\bm{s}_j},
\]
i.e., a linear combination of shuffle invariants of \(\mathcal{S}\).

Finally, it follows from Lemma~\ref{lem:anyshuffle} that for any \(X \in \shuffle(M_{1}, M_{2}, \ldots, M_{k})\), where as before \(M_1M_2 \cdots M_k \in \Omega\)  and whose Parikh vector is \(\bm{x} = \sum_{j=1}^p y_j \bm{s}_j\), the top right entry of \(X\) is equal to
\begin{eqnarray}\label{EqnShufInvars}
X_{1,n} &=& \Lambda_{\bm{x}} +  \sum_{1\leq i<j\leq k}  \alpha_{ij}[M_i,M_j]\ =\ \sum_{j=1}^p y_j \Lambda_{\bm{s}_j} +  \sum_{1\leq i<j\leq k}  \alpha_{ij}[M_i,M_j],
\end{eqnarray}
where each \(\alpha_{ij} \in \mathbb{Q}\) depends on the shuffle.

Furthermore, if a product of Heisenberg matrices is an \(\Omega\)-matrix and all matrix pairs share a common angle \(\gamma\) for their commutators, then shuffling the matrix product only modifies the top right element of the matrix by a real multiple of \(\exp(\iu\gamma)\). This drastically simplifies our later analysis.

\section{The Identity Problem for subsemigroups of \(\heisc\)}\label{SecIdentity}
In this section, we prove our main result. 

\begin{theorem}\label{thm:main}
Let \(G\subseteq \heisc\) be a finite set of matrices.
Then it is decidable in polynomial time if \(\bm{I}\in\langle G\rangle\).
\end{theorem}

The proof relies on analysing generators used in a product that results in an \(\Omega\)-matrix.
There are two distinct cases to consider:
either there is a pair of commutators with distinct angles, or else all commutators have the same angle.
The former case is considered in Lemma~\ref{lem:notallsamegamma} and the latter in Lemma~\ref{lem:samegamma}.
More precisely, 
we will prove that in the former case, the identity matrix is always in the generated semigroup and that the latter case reduces to deciding whether shuffle invariants reach the line defined by the angle of the commutator.

The two cases are illustrated in Figure~\ref{fig:gammalambda}.
On the left, is a depiction of the case where there are at least two commutators with different angles, \(\gamma_1\) and \(\gamma_2\).
We will construct a sequence of products where the top right element tends to \(r_1 \exp(\iu\gamma_1)\) with positive \(r_1\) and another product that tends to \(r_2 \exp(\iu\gamma_1)\) with negative \(r_2\).
This is achieved by changing the order of matrices whose commutator has angle \(\gamma_1\).
Similarly, we construct two sequences of products where the top right elements tend to \(r_3\exp(\iu\gamma_2)\) and \(r_4\exp(\iu\gamma_2)\), where \(r_3\) and \(r_4\) have the opposite signs.
Together these sequences ensure, that eventually, the top right elements do not lie in the same open half-planes.
On the right, is a depiction of the other case, where all commutators lie on \(\gamma\)-line.
In this case, the shuffle invariants of products need to be used to reach the line.

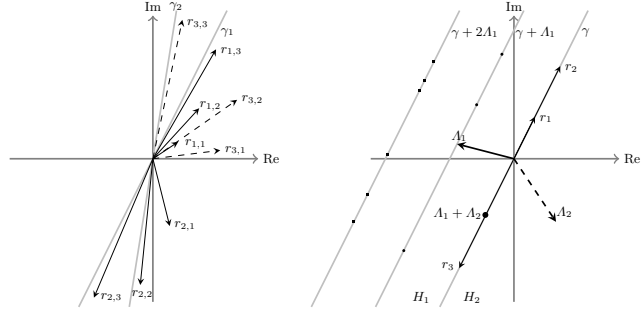
\begin{figure}[htb]
\centering
\scalebox{1}{
\begin{tikzpicture}

\begin{scope}[scale=.7, every node/.style={scale=0.7}]

\draw[->,gray,thick] (-3.4,0) -- (2.5,0);
\draw[->,gray,thick] (0,-3.4) -- (0,3.4);
\node[above] at (0,3.4) {$\Im$};
\node[right] at (2.5,0) {$\Re$};

\draw[thick,gray!50] (-1.76,-3.52) -- (1.76,3.52);
\node[right] at (1.5,3) {\(\gamma\)};

\draw[thick,gray!50] (-1.76-1.525,-3.52) -- (1.76-1.525,3.52);
\node[right] at (1.45-1.525,3) {\footnotesize\(\gamma+\Lambda_1\)};

\draw[thick,gray!50] (-1.76-2*1.525,-3.52) -- (1.76-2*1.525,3.52);
\node[right] at (1.45-2*1.525,3) {\footnotesize\(\gamma+2\Lambda_1\)};

\draw[-stealth,thick] (0,0) -- (-1.35,0.35);
\node[above] at (-1.3,0.3) {\(\Lambda_1\)};
\draw[-stealth,dashed,thick] (0,0) -- (1,-1.5);
\node[above] at (1.2,-1.5) {\(\Lambda_2\)};

\fill (-0.65cm-0.8pt,-1.3cm-1pt) circle (2pt);
\node[left] at (-0.65,-1.3) {\small\(\Lambda_1+\Lambda_2\)};

\draw[-stealth] (0,0) -- (0.5,1);
\node[right] at (0.5,0.95) {\small\(r_1\)};
\draw[-stealth] (0,0) -- (1.1,2.2);
\node[right] at (1.1,2.15) {\small\(r_2\)};
\draw[-stealth] (0,0) -- (-1.3,-2.6);
\node[left] at (-1.3,-2.55) {\small\(r_3\)};

\fill (0.5cm-1.35cm-1pt,0.95cm+0.35cm-0.5pt) circle (1pt);
\fill (1.1cm-1.35cm-1pt,2.15cm+0.35cm-0.5pt) circle (1pt);
\fill (-1.3cm-1.35cm+1pt,-2.55cm+0.35cm+0.5pt) circle (1pt);

\fill (0.8cm-2*1.35cm-1pt,1.6cm+2*0.35cm-1pt) rectangle (0.8cm-2*1.35cm+1pt,1.6cm+2*0.35cm+1pt);
\fill (-0.3cm-2*1.35cm-1pt,-0.6cm+2*0.35cm-1pt) rectangle (-0.3cm-2*1.35cm+1pt,-0.6cm+2*0.35cm+1pt);
\fill (-1.1cm-2*1.35cm-1pt,-2.2cm+2*0.35cm-1pt) rectangle (-1.1cm-2*1.35cm+1pt,-2.2cm+2*0.35cm+1pt);
\fill (-0.77cm-2*1.35cm-1pt,-1.54cm+2*0.35cm-1pt) rectangle (-0.77cm-2*1.35cm+1pt,-1.54cm+2*0.35cm+1pt);
\fill (0.46cm-2*1.35cm-1pt,0.92cm+2*0.35cm-1pt) rectangle (0.46cm-2*1.35cm+1pt,0.92cm+2*0.35cm+1pt);
\fill (0.58cm-2*1.35cm-1pt,1.16cm+2*0.35cm-1pt) rectangle (0.58cm-2*1.35cm+1pt,1.16cm+2*0.35cm+1pt);

\node at (-2.2,-3.3) {\(H_1\)};
\node at (-1,-3.3) {\(H_2\)};
\end{scope}

\begin{scope}[xshift=-6cm,scale=.7, every node/.style={scale=0.7}]

\draw[->,gray,thick] (-3.4,0) -- (2.5,0);
\draw[->,gray,thick] (0,-3.4) -- (0,3.4);
\node[above] at (0,3.4) {$\Im$};
\node[right] at (2.5,0) {$\Re$};

\draw[thick,gray!50] (-1.76,-3.52) -- (1.76,3.52);
\node[right] at (1.5,3) {\(\gamma_1\)};

\draw[thick,gray!50] (-0.56,-3.52) -- (0.56,3.52);
\node[above] at (0.55,3.4) {\(\gamma_2\)};

\draw[-stealth] (0,0) -- (0.6,0.4);
\node[right] at (0.65,0.35) {\small\(r_{1,1}\)};
\draw[-stealth] (0,0) -- (1.1,1.2);
\node[right] at (1,1.25) {\small\(r_{1,2}\)};
\draw[-stealth] (0,0) -- (1.5,2.6);
\node[right] at (1.5,2.55) {\small\(r_{1,3}\)};

\draw[-stealth] (0,0) -- (0.4,-1.6);
\node[right] at (0.4,-1.6) {\small\(r_{2,1}\)};
\draw[-stealth] (0,0) -- (-0.3,-3);
\node[below] at (-0.25,-3) {\small\(r_{2,2}\)};
\draw[-stealth] (0,0) -- (-1.4,-3.3);
\node[right] at (-1.35,-3.3) {\small\(r_{2,3}\)};

\draw[-stealth,dashed] (0,0) -- (1.6,0.2);
\node[right] at (1.6,0.2) {\small\(r_{3,1}\)};
\draw[-stealth,dashed] (0,0) -- (2,1.4);
\node[right] at (2,1.4) {\small\(r_{3,2}\)};
\draw[-stealth,dashed] (0,0) -- (0.7,3.3);
\node[right] at (0.7,3.2) {\small\(r_{3,3}\)};

\end{scope}

\end{tikzpicture}}
\caption{\label{fig:gammalambda} Illustrations of Lemma~\ref{lem:notallsamegamma} and Lemma~\ref{lem:samegamma}. Left shows two lines defined by two different commutators and how the values \(r_{1,\ell}\) and \(r_{2,\ell}\) tend to \(\gamma_1\)-line in opposite directions, while \(r_{3,\ell}\) tends to \(\gamma_2\)-line (\(r_{4,\ell}\) is omitted for clarity). Eventually, they are not all within the same closed half-plane. Right shows that if there is only one shuffle invariant, say \(\Lambda_1\), then all reachable values are on lines parallel to the \(\gamma\)-line, namely, \(\gamma+k\Lambda_1\) for \(k>0\). But if there exists \(\Lambda_2\) in the opposite half-plane, then the \(\gamma\)-line itself is reachable.}
\end{figure}

\begin{lemma}\label{lem:notallsamegamma}
Let $G=\{G_1,\ldots,G_t\}\subseteq \heisc$, where each \(G_i\) is non-redundant.
Suppose there exist \(M_1, M_2, M_3, M_4\in G\) such that \([M_{1},M_{2}]\not\gammaeq[M_{3},M_{4}]\). Then $\bm{I}\in \langle G\rangle$.
\end{lemma}
\begin{proof}
Since all \(M_i\)'s are non-redundant, there exists a product \(M_1 M_2 M_3 M_4 \cdots M_k\in \Omega\), where \(k\geq 4\) and every \(M_i\) is in \(G\). 
We prove the lemma by considering four products \(M_{1}^\ell M_{2}^\ell (M_{3}M_{4}X)^\ell\), \(M_{2}^\ell M_{1}^\ell (M_{3}M_{4}X)^\ell\), \(M_{3}^\ell M_{4}^\ell (M_{1}M_{2}X)^\ell\) and \(M_{4}^\ell M_{3}^\ell (M_{1}M_{2}X)^\ell\), where
\(X=\prod_{i=5}^{k}M_{i}\).
Denote 
\(\psi(M_i) = \left(\veca_i,\vecb_i,c_i\right)\) for \(i=1,\ldots,4\) and \(\psi(X) = \left(\bm{x}_1,\bm{x}_2,x_3\right)\).

Note that our assumption on the matrices \(M_1, M_2, M_3, M_4\) implies that there are at least four distinct matrices in the product \(M_1M_{2}\cdots M_{k}\in \Omega\) justifying the above notation.
Indeed, if there were at most three, the assumption would not hold.
Assume that \(X\) defined as above is the empty product and \(M_i=M_j\) for some \(i,j=1,2,3,4\) and \(i\neq j\).
For the sake of clarity, we will assume that \(i=1\) and \(j=3\). The other cases are analogous.
Now, \([M_1,M_2]=\veca_1^T\vecb_2-\veca_2^T\vecb_1\) and \([M_1,M_4]=\veca_1^T\vecb_4-\veca_4^T\vecb_1\), and since \(M_1M_2M_1M_4\in \Omega \), we also have
\(2\veca_1+\veca_2+\veca_4 = \vecz\) and \(2\vecb_1+\vecb_2+\vecb_4 = \vecz\).
That is, now 
\[
[M_1,M_4]=\veca_1^T\vecb_4-\veca_4^T\vecb_1=-\veca_1^T(2\vecb_1+\vecb_2)+(2{\veca_1+\veca_2})^T\vecb_1 
=\veca_2^T\vecb_1-\veca_1^T\vecb_2 = -[M_1,M_2],
\]
which implies \( [M_1,M_4] \gammaeq [M_1,M_2]\) contrary to our assumption.

Let $\gamma_1$ and $\gamma_2$ be the angles of \([M_1,M_2]\) and \([M_3,M_4]\), respectively, where \(\gamma_1 \neq \gamma_2\) by assumption. We show that, in the limit as \(\ell\) tends to infinity, the angles of the top-right entries in the first two products tend to \(\gamma_1\), but they approach \(\gamma_1\)-line (that is, the line defined by the vector \(r\exp(\iu\gamma_1)\), where \(r\) is a positive real number) from opposite direction, i.e., one approaches \(r\exp(\iu\gamma_1)\) and the other \(-r\exp(\iu\gamma_1)\).
The same holds for the last two products and the angle \(\gamma_2\). See the left half of Figure~\ref{fig:gammalambda} for illustration.

Let us consider the product \(M_1^\ell M_2^\ell (M_3M_4X)^\ell\) for some $\ell\in\N$.
By Lemma~\ref{lem:purewithcommutators},
\begin{multline*}
(M_1^\ell M_2^\ell (M_3M_4X)^\ell)_{1,n}  = \\ \ell \left(\sum_{i=1}^4c_i+x_3 
-\frac{1}{2}(\veca_1^T\vecb_1+\veca_2^T\vecb_2+(\veca_3+\veca_4+\bm{x}_1)^T(\vecb_3+\vecb_4+\bm{x}_2)\right) + \frac{\ell^2}{2}[M_1,M_2].
\end{multline*}
That is, the coefficient of \(\ell^2\) is \(\frac{1}{2}[M_1,M_2]\).
Similarly by Lemma~\ref{lem:purewithcommutators}, the coefficients of $\ell^2$ in the top right elements of $M_2^\ell M_1^\ell (M_3M_4X)^\ell$, $M_3^\ell M_4^\ell (M_1M_2X)^\ell$ and $M_4^\ell M_3^\ell (M_1M_2X)^\ell$ are
\(
\frac{1}{2}[M_2,M_1]\), \(\frac{1}{2}[M_3,M_4]\), \(\frac{1}{2}[M_4,M_3]
\),
respectively.
Let 
\begin{align*}
[M_1,M_2] = r_1\exp(\iu\gamma_1), &&
[M_2,M_1] = r_2\exp(\iu\gamma_2), \\
[M_3,M_4] = r_3\exp(\iu\gamma_3),  &&
[M_4,M_3] = r_4\exp(\iu\gamma_4).    
\end{align*}
It is convenient to consider these complex numbers as two-dimensional vectors.
Recall that commutator is antisymmetric and hence \(r_1 = -r_2\), \(\gamma_1=\gamma_2\) and \(r_3 = -r_4\), \(\gamma_3 = \gamma_4\).
By our assumption, \([M_1,M_2]\not\gammaeq [M_3,M_4]\) and thus \(\gamma_1\neq \gamma_3\).
It follows that the four vectors are not contained in any \emph{closed} half-plane.
Indeed, \(r_1\exp(\iu\gamma_1)\) and \(-r_1\exp(\iu\gamma_1)\) define two closed half-planes, say \(H_1\) and \(H_2\). Any \emph{closed} half-plane that contains both \(r_1\exp(\iu\gamma_1)\) and \(-r_1\exp(\iu\gamma_1)\) must be equal to either \(H_1\) or \(H_2\).
As \(r_3 = -r_4\), either \(r_3\exp(\iu\gamma_3)\) or \(-r_3\exp(\iu\gamma_3)\) is \emph{not} in that half-plane.

Let us express the top right elements as functions of power \(\ell\):
\begin{align*}
(M_1^{\ell} M_2^{\ell} (M_3M_4X)^{\ell})_{1,n}=r_{1,\ell}\exp(\iu\gamma_{1,\ell}), &&
(M_2^{\ell} M_1^{\ell} (M_3M_4X)^{\ell})_{1,n}=r_{2,\ell}\exp(\iu\gamma_{2,\ell}), \\
(M_3^{\ell} M_4^{\ell} (M_1M_2X)^{\ell})_{1,n}=r_{3,\ell}\exp(\iu\gamma_{3,\ell}), &&
(M_4^{\ell} M_3^{\ell} (M_1M_2X)^{\ell})_{1,n}=r_{4,\ell}\exp(\iu\gamma_{4,\ell}),
\end{align*}
where \(r_{1,\ell},r_{2,\ell},r_{3,\ell},r_{4,\ell}\in\R\) and \(\gamma_{1,\ell},\gamma_{2,\ell},\gamma_{3,\ell},\gamma_{4,\ell}\in[0,\pi)\).

Since \(r_1\exp(\iu\gamma_1)\), \(r_2\exp(\iu\gamma_1)\), \(r_3\exp(\iu\gamma_3)\), and \(r_4\exp(\iu\gamma_3)\) are the coefficients that multiply \(\frac{1}{2}\ell^2\) in the formula for the top-right entry in the above products, and \(\ell^2\) is the highest power that appears there, we conclude that
\[
\lim_{\ell\to\infty} \gamma_{1,\ell} = \lim_{\ell\to\infty} \gamma_{2,\ell} = \gamma_1
\qquad\text{and}\qquad
\lim_{\ell\to\infty} \gamma_{3,\ell} = \lim_{\ell\to\infty} \gamma_{4,\ell} = \gamma_3.
\]
Moreover, for sufficiently large \(\ell\), \(r_{i,\ell}\) and \(r_{i+1,\ell}\) have opposite signs, where \(i=1,3\).

Recall that \(r_1\exp(\iu\gamma_1)\), \(r_2\exp(\iu\gamma_2)\), \(r_3\exp(\iu\gamma_3)\), and \(r_4\exp(\iu\gamma_4)\) do not lie in the same closed half-plane.
It follows that, for sufficiently large \(\ell\), the vectors \(r_{1,\ell}\exp(\iu\gamma_{1,\ell})\), \(r_{2,\ell}\exp(\iu\gamma_{2,\ell})\), \(r_{3,\ell}\exp(\iu\gamma_{3,\ell})\), and \(r_{4,\ell}\exp(\iu\gamma_{4,\ell})\) also do not lie in the same closed half-plane. See the left half of Figure~\ref{fig:gammalambda} for illustration.

Since they do not belong to the same closed half-plane, it follows from the Fundamental Theorem of Linear Inequalities that the vector \(-r_{4,\ell}\exp(\iu\gamma_{4,\ell})\) can be expressed as a nonnegative linear combination of \(r_{1,\ell}\exp(\iu\gamma_{1,\ell})\), \(r_{2,\ell}\exp(\iu\gamma_{2,\ell})\), and \(r_{3,\ell}\exp(\iu\gamma_{3,\ell})\). For example, see Theorem 7.1 in \cite{Schrijver98}, where we set \(a_1 = r_{1,\ell}\exp(\iu\gamma_{1,\ell})\), \(a_2 = r_{2,\ell}\exp(\iu\gamma_{2,\ell})\), \(a_3 = r_{3,\ell}\exp(\iu\gamma_{3,\ell})\)), and \(b = -r_{4,\ell}\exp(\iu\gamma_{4,\ell})\).

Therefore, for for sufficiently large \(\ell\geq 1\), there are \(x_1,x_2,x_3,x_4\in\N\), not all of which are zero, such that
\[
x_1r_{1,\ell}\exp(\iu\gamma_{1,\ell})+x_2r_{2,\ell}\exp(\iu\gamma_{2,\ell})+x_3r_{3,\ell}\exp(\iu\gamma_{3,\ell})+x_4r_{4,\ell}\exp(\iu\gamma_{4,\ell})=0.
\]
This implies that
\[
\left(M_1^{\ell} M_2^{\ell} (M_3M_4X)^{\ell}\right)^{x_1}
\left(M_2^{\ell} M_1^{\ell} (M_3M_4X)^{\ell}\right)^{x_2} 
\left(M_3^{\ell} M_4^{\ell} (M_1M_2X)^{\ell}\right)^{x_3}
\left(M_4^{\ell} M_3^{\ell} (M_1M_2X)^{\ell}\right)^{x_4}
\]
equals the identity matrix \(\bm{I}\), 
finishing the proof. 
\end{proof}

It remains to consider the case when the angles of commutators coincide for each pair of non-redundant matrices.
Our aim is to prove that, under this condition, it is decidable whether the identity matrix is in the generated semigroup.

\begin{lemma}\label{lem:samegamma}
Let \(G=\{G_1,\ldots,G_t\}\subseteq \heisc\) be a set of non-redundant matrices such that the angle of commutator \([G_{i},G_{i'}]\) is \(\gamma\) for all \(1 \leq i, i' \leq t\), then we can determine in polynomial time if $\bm{I}\in \langle G\rangle$.
\end{lemma}
\begin{proof}
Let \(\{\bm{s}_1, \ldots, \bm{s}_p\} \subseteq \mathbb{N}^t\) be the set of minimal solutions to the SLDEs for \(G\) giving zeros in \(\bm{a}\) and \(\bm{b}\) elements. Each \(\bm{s}_j\) induces a shuffle invariant \(\Lambda_{\bm{s}_j} \in \GR\) as explained in Definition~\ref{shufinvdef}.

Consider a product \(X = M_{1} \cdots M_{k} \in\Omega\), where each \(M_{i} \in G\).  Let \(\bm{x} = (m_1, m_2, \ldots, m_{t}) \in \mathbb{N}^{t}\) be the Parikh vector of the number of occurrences of each matrix from \(G\) in product \(X\). Since \(X \in \Omega\), we have
\(\bm{x} = \sum_{j = 1}^{p} y_j \bm{s}_j\), where each \(y_j \in \mathbb{N}\). Notice that \(X \in \shuffle(G_1^{m_1}, \ldots, G_t^{m_t})\). Hence, by Equation (\ref{EqnShufInvars}), we have
\begin{equation}\label{eq:X1n}
X_{1,n} = \Lambda_{\bm{x}} +  \sum_{1\leq i<j\leq k}  \alpha_{ij}[M_i,M_j] = \sum_{j=1}^p y_j \Lambda_{\bm{s}_j} +  r \exp(\iu\gamma),
\end{equation}
where \(\alpha_{ij}\in \mathbb{Q}\) and \(r\in \mathbb{R}\) depend on the shuffle. In other words, any shuffle of the product \(X\) will change the top right entry \(X_{1,n}\) by a real multiple of \(\exp(\iu\gamma)\).

Let \(H_1, H_2\) be the two \emph{open} half-planes of the complex plane induced by \(\exp(\iu\gamma)\), that is, the union  \(H_1\cup H_2\) is the complement of the \(\gamma\)-line; thus \(0\not\in H_1\cup H_2\).
We now prove the following claim.

\begin{claim}
If \(\{\Lambda_{\bm{s}_1}, \ldots, \Lambda_{\bm{s}_p}\} \subseteq H_1\) or \(\{\Lambda_{\bm{s}_1}, \ldots, \Lambda_{\bm{s}_p}\} \subseteq H_2\), then $\bm{I}\notin \langle G\rangle$.
\end{claim}
\begin{proof}
Assume that \(\{\Lambda_{\bm{s}_1}, \ldots, \Lambda_{\bm{s}_p}\} \subseteq H_1\), renaming \(H_1, H_2\) if necessary. Assume that there exists some product \(X = X_1 X_2 \cdots X_k\) equal to the identity matrix, where \(k > 0\) and \(X_j \in G\). Then since \(X \in \Omega\), we see from Equation~\eqref{eq:X1n} that 
\(
X_{1,n} = \sum_{j=1}^p y_j \Lambda_{\bm{s}_j} + r \exp(\iu\gamma)
\), where \(r\in \mathbb{R}\).

Clearly, \(\sum_{j=1}^p y_j \Lambda_{\bm{s}_j} \in H_1\), and since \(y_j\neq 0\) for at least one \(j\), we have \(\sum_{j=1}^p y_j \Lambda_{\bm{s}_j} \neq 0\). Now, since \(r \exp(\iu\gamma)\) is on the \(\gamma\)-line, which is the boundary of \(H_1\), the value  \(X_{1,n}\) belongs to \(H_1\) and cannot equal zero. This contradicts the assumption that \(X\) is the identity matrix.
\end{proof}

So, if \(\{\Lambda_{\bm{s}_1}, \ldots, \Lambda_{\bm{s}_p}\}\) is not fully contained in either \(H_1\) or \(H_2\), then we cannot reach the identity matrix. Otherwise, there are two possibilities: (1) there exists some \(\Lambda_{\bm{s}_j} \in \GR\) such that the angle of \(\Lambda_{\bm{s}_j}\) is equal to \(\gamma\) (in which case such a \(\Lambda_{\bm{s}_j}\) lies on the line defined by \(\exp(\iu\gamma)\)), or (2) there exist \(\Lambda_{\bm{s}_i}, \Lambda_{\bm{s}_j}\), for \(1 \leq i < j \leq p\), such that \(\Lambda_{\bm{s}_i}\) and \(\Lambda_{\bm{s}_j}\) lie in different open half-planes, say \(\Lambda_{\bm{s}_i} \in H_1\) and \(\Lambda_{\bm{s}_j} \in H_2\).
Note that in the second case, there exist \(x, y \in \mathbb{N}\) such that\(x\Lambda_{\bm{s}_i} + y\Lambda_{\bm{s}_j} = r\exp(\iu\gamma)\) for some \(r \in \mathbb{R}\) since \(\Lambda_{\bm{s}_i}, \Lambda_{\bm{s}_j}\) and the commutators that define the \(\gamma\)-line have rational components. It means that in both cases there exist \(z_1, \ldots, z_p \in \mathbb{N}\) such that
\begin{equation}\label{eq:gline}
    \sum_{j=1}^{p}z_j\Lambda_{\bm{s}_j} = r\exp(\iu\gamma) \quad\text{for some}\quad r \in \mathbb{R}.
\end{equation}

Consider a tuple \(z_1, \ldots, z_p \in \mathbb{N}\) that satisfies Equation~(\ref{eq:gline}) and a product \(T = T_1 \cdots T_k\in \Omega\), such that each \(T_j\in G\) and whose Parikh vector is equal to \(\sum_{j=1}^{p}z_j\bm{s}_j\). It follows from Equation~\eqref{eq:X1n} that
\[
T_{1,n} = \sum_{j=1}^{p}z_j\Lambda_{\bm{s}_j} + r'\exp(\iu\gamma) = r \exp(\iu\gamma) + r'\exp(\iu\gamma),
\]
where \(r, r' \in \mathbb{R}\) and shuffles of such a product change only \(r'\).
Now, we have two possibilities.
\begin{itemize}
\itemindent=\leftmargin
    \item[\textbf{Case 1:}] There is a product \(T\), as defined above, that contains a pair of non-commuting matrices.
    
    \item[\textbf{Case 2:}] Any such product \(T = T_1 \cdots T_k\) consists only of commuting matrices from \(G\).
\end{itemize}

We consider each case separately, but firstly we show how to decide in \emph{polynomial time} which of them holds, that is, whether the \(\gamma\)-line can be reached by a pair of non-commuting matrices from \(G\).
To do this, we need to decide if there exist a shuffle invariant that lies on the \(\gamma\)-line and which uses a pair on non-commuting matrices. This may be determined by deciding the solvability of a polynomially sized set of Non-homogeneous Systems of Linear Diophantine Equations (NSLDEs), as explained below. We will show that determining if such a NSLDE has a solution can be done in polynomial time, and therefore the above property can be decided in \(\P\). We now outline the process.

Let \(v \in \GR\) be a complex number on the \(\gamma\)-line with rational components. Such \(v\) can be computed in polynomial time, e.g., it can be any non-zero commutator. We define its \emph{vectorization} as \(\vect(v) = (\Re(v), \Im(v))\), i.e., splitting into real and imaginary components.
The value \(v^{\perp} \in \GR\) whose vectorization \(\vect(v^{\perp})\) is perpendicular to \(\vect(v)\) in the complex plane is defined as \(v^{\perp} = \exp(\frac{i\pi}{2})v\) noting that \(\vect(v^{\perp})\) has rational components such that \(\vect(v^{\perp}) = (-\Im(v), \Re(v))\). Now, if all shuffle invariants are \emph{not} contained within \(H_1\) (or else \(H_2\)) then there exists two shuffle invariants, say \(\Lambda_{\bm{s}_1}, \Lambda_{\bm{s}_2} \in \GR\) such that \(\vect(\Lambda_{\bm{s}_1})\cdot \vect(v^{\perp}) > 0\) and \(\vect(\Lambda_{\bm{s}_2})\cdot \vect(v^{\perp}) < 0\), where~\(\cdot\) denotes the dot product, or else there exists a shuffle invariant \(\Lambda_{\bm{s}_3} \in \GR\) such that \(\vect(\Lambda_{\bm{s}_1})\cdot \vect(v^{\perp}) = 0\), thus \(\Lambda_{\bm{s}_3}\) is on the \(\gamma\)-line. 

Recall that \(G = \{G_1,\ldots, G_t\}\) and define \(\bm{y} = (c_1-\frac{1}{2}\bm{a}_1^T\bm{b}_1, \ldots, c_t-\frac{1}{2}\bm{a}_t^T\bm{b}_t)^T \in \GR^t\). Consider the vector \(\bm{z} = \Re(v^\perp)\Re(\bm{y}) + \Im(v^\perp)\Im(\bm{y}) \in \mathbb{Q}^t\). Now, let \(\bm{x}\) be the Parikh vector of some product of matrices from \(G\) that gives an \(\Omega\)-matrix. Then for  \(\bm{x}\) and its shuffle invariant \(\Lambda_{\bx} \in \GR\), we have that \(\bm{z}^T\bm{x} = \vect(\Lambda_{\bm{x}})\cdot \vect(v^{\perp})\).

We will now derive a set of NSLDEs, to allow us to determine if it is possible to find a product of matrices, not all of which commute, whose shuffle invariant lies on the \(\gamma\)-line. We require a non-homogeneous system in order to enforce that the solution is non-commuting.

We thus consider the following system of linear Diophantine equations \(A\bx\geq\bm{b}\), which we will now define. We may assume that not all matrices in generator set \(G = \{G_1, \ldots, G_t\}\) commute, since we will deal with this subcase later. Let us therefore consider a pair of non-commuting matrices \(G_i, G_j \in G\). In the construction of our NSLDE, submatrices \(A_1\) and \(-A_1\) are used to enforce that \(\bx\) is the Parikh vector of an \(\Omega\)-matrix \(M_{\bx} = M_1 M_2 \cdots M_k\), submatrices \(A_2, -A_2\) are used to determine that the shuffle invariant \(\Lambda_{\bx}\) is on the \(\gamma\)-line, and \(A_{i,j}\) is used to ensure that there exists two non-commuting matrices in the product \(M_1 M_2 \cdots M_k\), namely matrices \(G_i\) and \(G_j\). We will formulate such an NSLDE for each pair of non-commuting matrices.
Let us define \(T_{i,j}\bx\geq\bm{b}\) for every pair of non-commuting matrices \(G_i, G_j\) as follows:
\begin{equation}\label{eq:Tij}
T_{i,j} = \begin{pmatrix} A_1 \\ -A_1 \\ A_2 \\ -A_2 \\ A_{i,j} \end{pmatrix},
A_1 =\begin{pmatrix}
\Re(\veca_1) & \Re(\veca_2) & \cdots & \Re(\veca_t) \\
\Im(\veca_1) & \Im(\veca_2) & \cdots & \Im(\veca_t) \\
\Re(\vecb_1) & \Re(\vecb_2) & \cdots & \Re(\vecb_t) \\
\Im(\vecb_1) & \Im(\vecb_2) & \cdots & \Im(\vecb_t) \\
\end{pmatrix},
A_2 =\begin{pmatrix}
\bm{z}(1) \\ \bm{z}(2) \\ \vdots \\ \bm{z}(t)
\end{pmatrix}^T,
A_{i,j} =\begin{pmatrix}
\bm{e}_i^T \\ \bm{e}_j^T
\end{pmatrix}, 
\end{equation}
where \(\bx\in\N^t\), \(\bm{b} = (\bm{0}^{4(n-2)}, \bm{0}^{4(n-2)}, 0, 0, 1, 1)^T\); noting that \(T_{i,j} \in \mathbb{Q}^{(8(n-2)+4) \times t}\). Here, \(\bm{e}_i \in \{0, 1\}^t\) denotes the \(i\)'th standard (column) basis vector, i.e., the all zero vector except \(\bm{e}_i(i) = 1\), and similarly for \(\bm{e}_j\). A solution \(\bm{x}\in\N^t\) to this NSLDE implies that \(A_1 \bx \geq \bm{0}^{4(n-2)}\) and \(-A_1 \bx \geq \bm{0}^{4(n-2)}\), thus \(A_1 \bx = \bm{0}^{4(n-2)}\), and therefore \(G_1^{\bx(1)} \cdots G_t^{\bx(t)} \in \Omega\), which was the first property that we wished to enforce. Secondly, we see that \(A_2\bx \geq 0\) and  \(-A_2\bx \geq 0\) implies that \(\Lambda_{\bm{x}}\) is orthogonal to \(v^{\perp}\), i.e., the corresponding shuffle invariant is on the \(\gamma\)-line. Finally, \(A_{i,j}\bx \geq (1,1)^T\) implies that matrix \(G_i\) was used at least once, and matrix \(G_j\) was used at least once. Therefore, if the above system has a solution, then there is a product containing non-commuting matrices \(G_i, G_j\) that gives an \(\Omega\)-matrix, and whose shuffle invariant lies on the \(\gamma\)-line.

Note that the vector \(\bm{b}\) has non-negative coordinates, namely, \(0\)s and \(1\)s. Therefore, by Lemma~\ref{Ptime}, we can determine in polynomial time whether the system \(T_{i,j}\bx\geq\bm{b}\) has an integer solution.
There are \(t(t-1)/2\) many pairs of \(i,j\) that we need to check. Hence we can determine in polynomial time if the \(\gamma\)-line can be reached by using a pair of non-commuting matrices.

\medskip
Now we are ready to consider the \textbf{Cases 1} and \textbf{2} defined above.
\medskip

\textbf{Case 1.} We will show that in this case the identity matrix belongs to \(\langle G \rangle\).
Indeed, consider a shuffle \(T' = N_1N_2X' \in \shuffle(T_1, \ldots, T_k)\), where \(N_1\in G\) and \(N_2\in G\) do not commute and \(X'\) is the product of the remaining matrices in any order. We observe that Lemma~\ref{lem:purewithcommutators} implies 
\begin{align*}
(N_1^{\ell_1}N_2^{\ell_1}X'^{\ell_1})_{1,n} &= \ell_1 r\exp(\iu\gamma) + \frac{\ell_1^2}{2} [N_1, N_2]= \ell_1 r\exp(\iu\gamma) + \frac{\ell_1^2}{2} r'\exp(\iu\gamma)   
\quad \text{and}\\
(N_2^{\ell_2}N_1^{\ell_2}X'^{\ell_2})_{1,n} &= \ell_2 r\exp(\iu\gamma) + \frac{\ell_2^2}{2} [N_2, N_1]= \ell_2 r\exp(\iu\gamma) - \frac{\ell_2^2}{2} r'\exp(\iu\gamma),
\end{align*}
for some \(0\neq r' \in \mathbb{R}\). We then notice that 
\begin{multline*}
\left((N_1^{\ell_1}N_2^{\ell_1}X'^{\ell_1})^{d_1}(N_2^{\ell_2}N_1^{\ell_2}X'^{\ell_2})^{d_2}\right)_{1,n} 
=  \\
d_1\left(\ell_1r\exp(\iu\gamma) + \frac{\ell_1^2}{2}r'\exp(\iu\gamma)\right) + d_2\left(\ell_2r\exp(\iu\gamma) - \frac{\ell_2^2}{2}r'\exp(\iu\gamma)\right).
\end{multline*}
Now,
\begin{align*}
d_1\left(\ell_1r\exp(\iu\gamma) + \frac{\ell_1^2}{2}r'\exp(\iu\gamma)\right) + d_2\left(\ell_2r\exp(\iu\gamma) - \frac{\ell_2^2}{2}r'\exp(\iu\gamma)\right) &= 0 \\
\Longleftrightarrow d_1(2\ell_1r + \ell_1^2r') + d_2(2\ell_2r - \ell_2^2r') &= 0
\\
\Longleftrightarrow\ d_1(2\frac{r}{r'}\ell_1 + \ell_1^2) + d_2(2\frac{r}{r'}\ell_2 - \ell_2^2) &= 0.
\end{align*}
By our assumption, the vectors \(r\exp(\iu\gamma)\) and \(r'\exp(\iu\gamma)\) have rational coordinates and the same angle~\(\gamma\). It follows that \(\frac{r}{r'}\in \mathbb{Q}\). 
Hence we may choose sufficiently large \(\ell_1, \ell_2 > 1\) such that \(2\frac{r}{r'}\ell_1 + \ell_1^2\) and \(2\frac{r}{r'}\ell_2 - \ell_2^2\) have different signs, and then integers \(d_1, d_2 > 1\) can be chosen that satisfy the above equation. This choice of \(\ell_1, \ell_2, d_1, d_2\) is then such that
\(
(N_1^{\ell_1}N_2^{\ell_1}X'^{\ell_1})^{d_1}(N_2^{\ell_2}N_1^{\ell_2}X'^{\ell_2})^{d_2} = \bm{I}
\)
as required. Thus if such non-commuting matrices are present, we can reach the identity.

\textbf{Case 2.} Finally, we consider the case when all matrices which can be used to reach the $\gamma$-line commute. 
Since we also assumed at the beginning of the proof that \(\{\Lambda_{\bm{s}_1}, \ldots, \Lambda_{\bm{s}_p}\}\) is not contained in \(H_1\) or \(H_2\), we have two cases:
\begin{itemize}
\itemindent=\leftmargin
    \item[\textbf{Case A:}] There exist \(\Lambda_{\bm{s}_i}, \Lambda_{\bm{s}_j}\) such that \(\Lambda_{\bm{s}_i}\in H_1\) and \(\Lambda_{\bm{s}_j}\in H_2\), or
    
    \item[\textbf{Case B:}] \(\{\Lambda_{\bm{s}_1}, \ldots, \Lambda_{\bm{s}_p}\}\subseteq H_1 \cup \gamma\text{-line}\).

    (Actually, there is a third possibility that \(\{\Lambda_{\bm{s}_1}, \ldots, \Lambda_{\bm{s}_p}\}\subseteq H_2 \cup \gamma\text{-line}\) but it is similar to this case by symmetry.)
\end{itemize}

Note that we can determine in polynomial time which of these cases hold. Indeed, the first case holds if and only if there exist vectors \(\bm{x}, \bm{y}\in\N^t\) such that \(A_1 \bm{x} = \bm{0}^{4(n-2)}\), \(A_1 \bm{y} = \bm{0}^{4(n-2)}\), \(A_2\bm{x} > 0\), and \(A_2\bm{y} < 0\). To see this, take \(\bm{x}\in\N^t\) such that \(A_1 \bm{x} = \bm{0}^{4(n-2)}\) and \(A_2\bm{x} > 0\). Then \(\bm{x}\) can be written as \(\bm{x} = \sum_{j=1}^{p}z_j\bm{s}_j\) for some tuple \((z_1,\ldots,z_p)\in \N^p\). So we have \(A_2\bm{x} = \sum_{j=1}^{p}z_jA_2\bm{s}_j > 0\), which implies that \(A_2\bm{s}_i > 0\) for some \(i\) since all \(z_j\) are non-negative. Hence \(\Lambda_{\bm{s}_i}\in H_1\). All the other cases can be considered in a similar way, and
the existence of a solution to such system can be decided in polynomial time by the second part of Lemma~\ref{Ptime}.

Next, we show the following claim.
\begin{claim}
   In both cases, the set \(C = \{G_1,\ldots,G_{t'}\} \subseteq G\) of commuting matrices that can be used to reach the \(\gamma\)-line is computable in polynomial time. 
\end{claim}

\begin{proof}
Firstly, we show that in \textbf{Case A}, all matrices from \(G\) can be used in some product that is equal to an \(\Omega\)-matrix whose top-right entry lies on the \(\gamma\)-line.
Indeed, since we assume that each matrix \(M_k\) from \(G\) is non-redundant, there is some \(\bm{s}_i\) with non-zero \(k\)th coordinate. Now, the shuffle invariant \(\Lambda_{\bm{s}_i}\) can be paired with some \(\Lambda_{\bm{s}_j}\) from the other open half-plane to reach the \(\gamma\)-line.
Namely, there exist \(x, y \in \mathbb{N}\) such that \(x\Lambda_{\bm{s}_i} + y\Lambda_{\bm{s}_j} = r\exp(\iu\gamma)\) for some \(r \in \mathbb{R}\). Hence we can find a tuple \((z_1,\ldots,z_p)\in \N^p\) such that \(\sum_{j=1}^{p}z_j\Lambda_{\bm{s}_j} = r\exp(\iu\gamma)\), for some \(r \in \mathbb{R}\), and the vector \(\sum_{j=1}^{p}z_j\bm{s}_j\) has only non-zero coordinates. This gives us a product with Parikh vector \(\sum_{j=1}^{p}z_j\bm{s}_j\) that uses all matrices from \(G\) and reaches the \(\gamma\)-line.

On the other hand, in \textbf{Case B}, a matrix \(G_k\in G\) can be used in some product that reaches the \(\gamma\)-line if and only if the \(k\)th coordinate of some \(\bm{s}_i\), for which \(\Lambda_{\bm{s}_i}\) lies on the \(\gamma\)-line, is non-zero. In other words, precisely the following matrices can be used to reach the \(\gamma\)-line:
\begin{equation}\label{eq:gammaline}
    \{G_k \in G\ :\ \text{there is some } \bm{s}_i \text{ such that } \bm{s}_i(k)>0 \text{ and } \Lambda_{\bm{s}_i} = r\exp(\iu\gamma) \text{ for some } r \in \mathbb{R}\}.
\end{equation}

To show this, assume \(G_k\) belongs to a product that reaches the \(\gamma\)-line and let \(\bm{x}\in \N^t\) be its Parikh vector. We can write \(\bm{x} = \sum_{j=1}^{p}z_j\bm{s}_j\), for some \((z_1,\ldots,z_p)\in \N^p\). Note that \(\Lambda_{\bm{x}} = \sum_{j=1}^{p}z_j\Lambda_{\bm{s}_j} = r\exp(\iu\gamma)\) for some \(r \in \mathbb{R}\). Since we assumed \(\{\Lambda_{\bm{s}_1}, \ldots, \Lambda_{\bm{s}_p}\}\subseteq H_1 \cup \gamma\text{-line}\), it follows that \(z_j=0\) if \(\Lambda_{\bm{s}_j}\in H_1\). This means that only those \(\Lambda_{\bm{s}_j}\) that lie on the \(\gamma\)-line can appear in the linear combination \(\Lambda_{\bm{x}} = \sum_{j=1}^{p}z_j\Lambda_{\bm{s}_j}\). Since \(G_k\) appears in the product, we have \(\bm{x}(k) = \sum_{j=1}^{p}z_j\bm{s}_j(k) > 0\). Thus \(z_j\bm{s}_j(k) > 0\) for some \(j\). In particular, \(\bm{s}_j(k) > 0\) and \(z_j>0\), which implies that \(\Lambda_{\bm{s}_j}\) is on the \(\gamma\)-line. Therefore, \(G_k\) belongs to the set (\ref{eq:gammaline}).

Conversely, consider all matrices \(G_k \in G\) for which there is some \(\bm{s}_{i_k}\) with the property that \(\bm{s}_{i_k}(k)>0\) and \(\Lambda_{\bm{s}_{i_k}}\) lies on the \(\gamma\)-line. Let \(\bm{x}\) be the sum of these \(\bm{s}_{i_k}\). In this case, \(\Lambda_{\bm{x}}\) is on the \(\gamma\)-line, and we can construct a product with Parikh vector \(\bm{x}\) that reaches the \(\gamma\)-line and contains all matrices from the set (\ref{eq:gammaline}), namely, it suffices to take any product with Parikh vector \(\bm{x}\).

Note that the condition ``there is some \(\bm{s}_i\) such that \(\bm{s}_i(k)>0\) and \(\Lambda_{\bm{s}_i}\) is on the \(\gamma\)-line'' is equivalent to the following: there is some \(\bm{x}\in \N^t\) with \(\bm{x}(k)>0\) such that \(A_1\bm{x} = \bm{0}^{4(n-2)}\) and \(\Lambda_{\bm{x}}\) is on the \(\gamma\)-line. The implication in one direction is obvious. Suppose there is \(\bm{x}\in \N^t\) with \(\bm{x}(k)>0\) such that \(A_1\bm{x} = \bm{0}^{4(n-2)}\) and \(\Lambda_{\bm{x}}\) is on the \(\gamma\)-line. Then we can write \(\bm{x}\) as a sum \(\bm{x} = \sum_{j=1}^{p}z_j\bm{s}_j\), where \((z_1,\ldots,z_p)\in \N^p\). Since \(\bm{x}(k)>0\), there is \(j\) such that \(\bm{s}_j(k) > 0\). Also note that \(\Lambda_{\bm{s}_i}\) must be on the \(\gamma\)-line,  otherwise \(\Lambda_{\bm{x}}\) would not be on the \(\gamma\)-line because \(z_j>0\). Hence we can determine which matrices belong to (\ref{eq:gammaline}) by deciding if there is a solution to a system on linear equation similar to (\ref{eq:Tij}), with the exception that matrix \(A_{i,j}\) should be replaced with \(A_k = (\bm{e}_k^T)\) and vector \(\bm{b}\) has the form \(\bm{b} = (\bm{0}^{4(n-2)}, \bm{0}^{4(n-2)}, 0, 0, 1)^T\). This can be done in polynomial time by Lemma~\ref{Ptime}.
\end{proof}

So, in both cases, the set \(C = \{G_1,\ldots,G_{t'}\} \subseteq G\) of commuting matrices that can be used to reach the \(\gamma\)-line is computable in polynomial time.
To finish the proof of \textbf{Case 2}, note that by \eqref{commutingURval}, the top-right corner \(M_{1,n}\) of any \(M\in \langle C\rangle\cap\Omega\) can be expressed using the Parikh vector of the generators from \(C\).
This allows us to construct a new homogeneous system of linear Diophantine equations.
Let \(A\in \Q^{4(n-2)\times t'}\) be defined as in Equation \eqref{eq:matA} using only matrices present in \(C\), and let \(B=(c_1-\frac{1}{2}\veca_1^T\vecb_1,\ldots,c_{t'}-\frac{1}{2}\veca_{t'}^T\vecb_{t'})\).
We then construct a system
\(\begin{psmallmatrix}
A \\ B
\end{psmallmatrix}\bx=\bm{0}\),
where \(\bx\in \N^{t'}\) and \(\bm{0}\) is the \(t'\)-dimensional zero vector.
It is straightforward to see that if this system has a solution \(\bm{x}\), then 
\[G_1^{x(1)}G_2^{x(2)}\cdots G_{t'}^{x(t')}=\bm{I}.\]
By Lemma~\ref{Ptime} (see also \cite{Khachiyan80}), we can decide if such a system has a non-zero solution in polynomial time. 
\end{proof}

Lemmata~\ref{lem:notallsamegamma} and \ref{lem:samegamma} allow us to prove the main result, Theorem~\ref{thm:main}.
\begin{proofof}{Theorem~\ref{thm:main}}
Let \(G=\{G_1,\ldots,G_t\}\).
The first step is to remove all redundant matrices from~\(G\). Recall that we can check if a matrix is non-redundant by deciding if a system of non-homogeneous linear Diophantine equations, giving an \(\Omega\)-matrix, has a solution. This was explained at the beginning of Section~\ref{SecShufInv}.
It is decidable in polynomial time if there exists a solution to such a system.

It is obvious that if there is no solution to the system, then the identity matrix is not in the generated semigroup.
Indeed, there is no way to generate a matrix with zeroes in the \(\bm{a}\) and \(\bm{b}\) elements and thus \(\bm{I}\not\in\langle G\rangle\).

If there is at least one solution, then we calculate commutators \([G_i,G_j]\) for every pair \(G_i,G_j\) of non-redundant matrices from \(G\).
If there are at least two commutators with different angles, the identity matrix is in the semigroup by Lemma~\ref{lem:notallsamegamma}.
If all commutators have the same angle, we apply Lemma~\ref{lem:samegamma} to decide in polynomial time whether the identity matrix is in the semigroup.

Note that there are \(O(t^2)\) commutators to calculate. Hence the whole procedure runs in polynomial time. 
\end{proofof}

The decidability of the Identity Problem implies that the Subgroup Problem is also decidable.
That is, whether the semigroup generated by the generators \(G\) contains a non-trivial subgroup.
However, the decidability of the Group Problem, i.e., whether \(\langle G\rangle\) is a group, does not immediately follow.
Our result can be extended to show decidability of the Group Problem.
\begin{corollary}\label{cor:group}
It is decidable in polynomial time whether a finite set of matrices \(G\subseteq\heisc\) forms a group.
\end{corollary}
\begin{proof}
In order for \(\langle G\rangle\) to be a group, each element of \(G\) must have a multiplicative inverse. If there exist any redundant matrices in \(G\), then \(\langle G \rangle\) is not a group, since a redundant matrix cannot be part of a product giving even an \(\Omega\)-matrix. Checking if a matrix is redundant can be done in polynomial time (see Section~\ref{SecShufInv}).

Assuming then that all matrices are non-redundant, if there exist matrices \(M_1, M_2, M_3, M_4 \in G\) such that \([M_{1},M_{2}]\not\gammaeq[M_{3},M_{4}]\), then \(\bm{I} \in \langle G \rangle\) by Lemma~\ref{lem:notallsamegamma}, and in fact we can find a product of matrices equal to the identity matrix which contains each matrix from \(G\) (since all matrices are non-redundant, thus \(M_1\cdots M_k \in \Omega\) may be chosen to contain all matrices). Thus each matrix of \(G\) has a multiplicative inverse as required. 

Next, assume that all matrices in \(G\) are non-redundant and share a common commutator angle. As explained in the proof of Lemma~\ref{lem:samegamma}, we can compute in polynomial time a subset \(C\subseteq G\) of matrices that can be used in a product which is equal to an \(\Omega\)-matrix whose top-right element lies on the \(\gamma\)-line. Clearly, if \(C\neq G\), then \(\langle G \rangle\) is not a group. If \(C=G\) and \(G\) contains a pair of non-commuting matrices, then in the proof of Lemma~\ref{lem:samegamma} we can choose a product \(T = T_1 \cdots T_k\in \Omega\) in such a way that it includes all matrices from \(G\). Using the same idea we can construct a product that gives \(\bm{I}\) and uses all matrices from \(G\). Hence \(\langle G \rangle\) is a group.

Finally, we need to consider the case when \(C=G\) contains only commuting matrices. We can decide if there is a product that is equal to \(\bm{I}\) and uses all matrices from \(G\) by solving a homogeneous system of linear Diophantine equations.
Let \(A\in \Q^{4(n-2)\times t}\) be defined as in Section~\ref{SecShufInv}, and let \(A_2=(c_1-\frac{1}{2}\veca_1^T\vecb_1,\ldots,c_{t}-\frac{1}{2}\veca_{t}^T\vecb_{t})\).
Define system
\(\begin{psmallmatrix}
A \\ A_2
\end{psmallmatrix}\bx=\bm{0}^{4(n-2)+1}\),
where \(\bx\in \N^{t}\) and \(\bm{0}^{4(n-2)+1} = (0, \ldots, 0) \in \mathbb{N}^{4(n-2)+1}\). Solvability of this system implies that the identity matrix can be reached. Unfortunately, this does not guarantee that all matrices from \(G\) were used in such a product reaching \(\bm{I}\). We can however form the following non-homogeneous system of linear Diophantine equations  
\[
\begin{psmallmatrix}
A \\ -A \\ A_2 \\ -A_2 \\ I_t
\end{psmallmatrix}\bx\geq
\begin{psmallmatrix}
\bm{0}^{4(n-2)} \\ \bm{0}^{4(n-2)} \\ 0 \\ 0 \\ \bm{1}^t
\end{psmallmatrix},
\]
where \(\bm{1}^t = (1, \ldots, 1) \in \mathbb{N}^t\) and \(\bm{I}_t\) is the \(t \times t\) identity matrix. Solvability of this system is equivalent to the homogeneous system defined above, with the added constraint that \(\bm{I}_t \bx \geq \bm{1}^t\) implying that all matrices must be used at least once in such a solution, as required for \(\langle G \rangle\) to be a group. By Lemma~\ref{Ptime}, we can decide in polynomial time if the above system has a solution. 
\end{proof}

\section{Future research}
We believe that the techniques, and the general approach, presented in the previous sections can act as stepping stones for related problems.
In particular, consider the Membership Problem, i.e., where the target matrix can be any matrix rather than the identity matrix (Problem~\ref{prob:membership}).
Let \(M=\begin{psmallmatrix}
    1&\bm{m}^T_1&m_3 \\ \bm{0}&\bm{I}_{n-2}&\bm{m}_2\\0&\bm{0}^T&1
\end{psmallmatrix}\) be the target matrix and let \(G=\{G_1,\ldots,G_t\}\), where \(\psi(G_i)=(\bm{a}_i,\bm{b}_i,c_i)\).
Following the idea of Section~\ref{SecShufInv}, we can consider system \(A\bx=(\bm{m}_1,\bm{m}_2)\), where
\(\bx\in\N^t\).
This system is a non-homogeneous system of linear Diophantine equations that can be solved in \(\NP\).
The solution set is a union of two finite solution sets, \(S_0\) and \(S_1\).
The set \(S_0\) is the solutions to the corresponding homogeneous system that can be repeated any number of times as they add up to \(\bm{0}\) on the right-hand side.
The other set, \(S_1\), corresponds to reaching the vector \((\bm{m}_1,\bm{m}_2)\).
The matrices corresponding to the solutions in \(S_1\) have to be used exactly this number of times.

The techniques developed in Section~\ref{SecShufInv} allow us to manipulate matrices corresponding to solutions in \(S_0\) in order to obtain the desired value in the top right corner.
However, this is not enough as the main technique relies on repeated use of \(\Omega\)-matrices.
These can be interspersed with matrices corresponding to a solution in \(S_1\) affecting the top right corner in uncontrollable ways.

Consider then the Membership Problem for Heisenberg matrices of dimension two, that is for matrices of the form \(\begin{psmallmatrix}
    1&a\\0&1
\end{psmallmatrix}\), where \(a\in\C\).
These matrices commute.
This means that if we construct a non-homogeneous system of two linear Diophantine equations (one equation for the real components and another for the imaginary) similar to what we did in Section~\ref{SecShufInv}, we capture all solutions.
However, this problem is then \(\NP\)-complete as it is a variant of the \emph{integer programming} \cite{GJ79}.
Note that if the number of generators in the generating set is fixed, the problem is solvable in polynomial time by the result of \cite{BBC+96} as all matrices commute.

This motivates the following open problem:
\begin{problem}[Membership Problem for Complex Heisenberg Matrices]\label{prob:membershipHeis}
Let \(n\geq3\),  $S$ be a matrix semigroup generated by a finite set of
$n{\times}n$ matrices over \heisc and let \(M\in\heisc\).
Is $M$ in the semigroup, i.e., does $M\in S$ hold?
\end{problem}

\nocite{*}
\bibliographystyle{fundam}
\bibliography{references}


\end{document}